\documentclass{article}
\usepackage{amsmath} 
\usepackage{amssymb}
\usepackage{amsthm}
\usepackage{color}
\usepackage{fullpage}

\usepackage[bookmarks=true, unicode=true, pdftitle={GRANDPA: a blockchain finality gadget}, pdfauthor={Alistair Stewart},pdfkeywords={blockchain finality gadget, consensus, Polkadot, Substrate},pdfborder={0 0 0.5 [1 3]}]{hyperref}
\usepackage{url}
\usepackage[numbers]{natbib}

\usepackage{tabu} 

\PassOptionsToPackage{hyphens}{url}\usepackage{hyperref}

\newtheorem{theorem}{Theorem}[section]
\newtheorem{definition}[theorem]{Definition}
\newtheorem{lemma}[theorem]{Lemma}
\newtheorem{corollary}[theorem]{Corollary}
\newtheorem{proposition}[theorem]{Proposition}

\def\GST{\mathrm{GST}}

\long\def\com#1{}

\begin{document}

\title{GRANDPA: a Byzantine Finality Gadget}
\author{Alistair Stewart \\ {\tt stewart.al@gmail.com} \and Eleftherios Kokoris-Kogia \\ {\tt eleftherios.kokoriskogias@epfl.ch}}
\date{June 30, 2020}
\maketitle

\begin{abstract}
Classic Byzantine fault-tolerant consensus protocols forfeit liveness in the face of asynchrony in order to preserve safety, whereas most deployed blockchain protocols forfeit safety in order to remain live. 
In this work, we achieve the best of both worlds by proposing a novel abstractions called the \emph{finality gadget}.
A finality gadget allows for transactions to always optimistically commit but informs the clients that these transactions might be unsafe. As a result, a blockchain can execute transactions optimistically and only commit them after they have been sufficiently and provably audited.
In this work, we formally model the finality gadget abstraction, prove that it is impossible to solve it deterministically in full asynchrony (even though it is stronger than consensus) and provide a partially synchronous protocol which is currently securing a major blockchain. This way we show that the protocol designer can decouple safety and liveness in order to speed up recovery from failures. We believe that there can be other types of finality gadgets that provide weaker safety (e.g., probabilistic) in order to gain more efficiency and this can depend on the probability that the network is not in synchrony.
\end{abstract}
\section{Introduction}

Bitcoin~\cite{nakamoto08bitcoin} and its descendants~\cite{wood14ethereum,sasson2014zerocash} are cryptocurrencies that provide 
secure automated value exchange without the need of a central managing authority.
Instead a decentralized consensus protocol maintains a distributed public ledger
known as the \textit{blockchain}. To be able to rely on public ledger one needs to know that it has reached consensus on a certain block, i.e., when a block will not be reverted anymore, which we refer to as reaching finality.  
One of the challenges of Nakomoto-like consensus protocols is that they only satisfy eventual consensus, which only guarantees that an ever growing prefix of the chain will be agreed upon by all participants forever onward. The eventual consensus process generally takes tens of minutes and it only gives probabilistic guarantees (for a certain block number at a certain point in time). 

Unfortunately these guarantees only hold if the underlying network is well-connected and the client is able to find an uncensored source of information, two assumptions that
do not hold in adversarial environments~\cite{apostolaki16hijacking, gervais15tampering, heilman15eclipse}.
The underlying problem which enables these attacks is that first generation blockchain protocols do not consider finality (i.e., when will a block never be reverted) as a first class property, prioritising liveness instead.

An alternative to probabilistic finality is having \emph{provable finality} where anyone can be convinced of the finality of a block, regardless of being a consensus participants or actively following the network.
New generation protocols~\cite{kokoris16enhancing,decker16bitcoin,pass16hybrid} propose the complete opposite. They propose every block to be finalized one by one and to forfeit liveness when finality is not readily achievable. This gives provable finality immediately.
Unfortunately, these types of protocols inherit the shortcoming of classic consensus protocol on losing performance when many nodes are required to participate. Hence, they need to put a limit on the number of consensus participants which might lead to centralization. 


In this work we show the that the middle ground also merits exploration. The approach that we will take is similar to the approach that Ethereum plans to take with Casper the Friendly Finality Gadget (Casper FFG)[2].
We introduces and formalize the idea of lazy finality which is encapsulated in the abstraction of a \emph{finality gadget.} Separating the liveness of the consensus protocol from the finality of the blocks. This approach has three concrete benefits for the overlying blockchain protocol. 


A first benefit is that since consensus is not tied to liveness of the chain we can have optimistic execution.  As a result, the chain can grow before it is certain that blocks are valid. Later on, we can finalize blocks when we are sure they are correct i.e., all verification information is available. 
A second benefit is that we can make some (unsafe) progress when the network is unstable. This enables a fast recovery process when the network heals. Similarly, we can make progress in chain growth even when finlzation is slow, e.g., when we have many particiapnts thus promoting decentenralization. 



Third, a finality gadget can be deployed gradually and light clients can choose to consult it or follow the longest chain rule and ignore it, enabling light client heterogeneity. 
The light client that trust the gadget do not need to have the full chain or actively listen to the network. 
This can in turn enable scalability~\cite{avarikioti19divide} in an ecosystem of multiple chains (weather sharding~\cite{kokoris17omniledger,al18chainspace,androulaki18channels} or heterogeneous~\cite{zamyatin19sok}), where no single party receives or stores all the data in the system. 


In short, we fromalize the abstraction of a \emph{finality gadget} that runs along any block production mechanism (e.g., Nakamoto consensus) providing provable finality guarantees. We show that it is impossible to satisfy its properties with a deterministic asynchronous protocol. To circumvent this impossibility result, we introduce the GRANDPA finality gadget that works in a partially synchronous network model, in the presence of up to $1/3$ Byzantine actors. 

The combination of GRANDPA with a classic block production mechanism like GHOST~\cite{lewenberg15inclusive} results in the existing deployment of the polkadot network \footnote{\url{https://polkadot.network}} which provides fast finality under good network conditions and protects the clients without compromising the liveness when under attack. The implementation of GRANDPA is available on github \footnote{See \url{https://github.com/paritytech/finality-grandpa} and \url{https://github.com/paritytech/substrate/tree/master/client/finality-grandpa}}.

In summary we make the following contributions:
\begin{itemize}

\item Introduce the idea of lazy finality and instantiate it through a finality gadget abstraction

\item Prove that BFG is impossible in asynchrony and present GRANDPA

\end{itemize}



\section{Model, Definitions, and Impossibilities}

We want to formalise the notion of finality gadget to be a sub-protocol that can be deployed along any protocol with
eventual consensus and probabilistic finality and enhancing such protocol with provable finality.
To achieve this, we need to incorporate into the classic definition of Byzantine agreement 
the fact that we additionally have access to a protocol that would achieve eventual consensus if we did not affect it.

\subsection{Byzantine Agreement with a Consistency Oracle}
Consider a typical definition of a multi-values Byzantine agreement: 
We have a set of participants $V$, the majority of whom obey the protocol, but a constant fraction may be Byzantine, meaning they behave arbitrarily, e.g. provide false or inconsistent information or randomly go offline when they ought to be online.

\begin{definition} A protocol for {\em multi-valued Byzantine agreement} has a set of values $S$ and a set of voters $V$, a constant fraction of which may be Byzantine, for which each voter $v \in V$ starts with an initial value $s_v \in S$ and, in the end, decides a final value $f_v \in S$ such that the following holds:

\begin{itemize}
\item {\bf Agreement}: All honest voters decide the same value for $f_v$
\item {\bf Termination}: All honest voters eventually decide a value
\item {\bf Validity}: If all honest voters have the same initial value, then they all decide that value
\end{itemize}

\end{definition}

We can change this definition to assume that instead of having an initial value, all voters have access to an external protocol, an oracle for values, that achieves eventual consensus in that it returns the same value to all voters when called after some unknown time.

\begin{definition}
We say an oracle $A$ in a protocol is {\em eventually consistent} if it returns the same value to all participants after some unspecified time.
\end{definition}

\begin{definition} A protocol for the {\em multi-valued Byzantine finality gadget problem} has a set of values $S$, a set of voters $V$, a constant fraction of which may be Byzantine, for which each voter $v \in V$ has access to an eventually consistent oracle $A$ and, in the end, each voter decides a final value $f_v \in S$ such that the following holds:

\begin{itemize}
\item {\bf Agreement:} All honest voters decide the same value for $f_v$
\item {\bf Termination:} All honest voters eventually decide a value
\item {\bf Validity:} All honest voters decide a value that $A$ returned to some honest voter sometime.
\end{itemize}

\end{definition}

\paragraph{Impossibility of Deterministic Agreement with an Oracle.}\label{ssec:impossibility}
For the binary case, i.e. when $|S|=2$, the Byzantine finality gadget problem is reducible to Byzantine agreement. This does not hold for $|S| > 2$, because the definition of validity is stronger in our protocol. Note that it is impossible for multi-valued Byzantine agreement to make the validity condition require that we decide an initial value of some honest voter and tolerate more than a $1/|S|$ fraction of faults, since we may have a $1/|S|$ fraction of voters reporting each initial value and Byzantine voters can act honestly enough not to be detectable. For finality gadgets, this stronger validity condition is possible. A natural question is then weather the celebrated FLP~\cite{flp} impossibility holds for our stronger requirements.
Next, we show that an asynchronous, deterministic binary finality gadget is impossible, even with one fault. 
This means that the extra information voters have here, that $A$ will eventually agree for all voters, is not enough to make this possible.

\paragraph{Proof:}
The asynchronous binary fault tolerant agreement problem is as follows:

We have  number of voters which each have an initial $v_i$ in $\{0,1\}$

We may have one or more faulty nodes, which here means going offline at some point. Nodes have asynchronous communication - so any message arrives but we have no guarantee when it will.
The goal is to have all non-faulty nodes output the same $v$, which must be $0$ if all inputs $v_i$ are $0$ and $1$ if all are $1$.

Fischer, Lynch and Paterson\cite{flp} showed that this is impossible if there is one faulty node.

The binary fault-safe finality gadget problem is similar, except now there is an oracle $A$ that any node can call at any time with the following properties:

either $A$ always outputs $x$ in $\{0,1\}$ to all nodes at all times
or else there is an $x$ in $\{0,1\}$ and
for each node $i$, there is a $T_i$ such that when $i$ calls $A$ before $T_i$. it gives $x$ but if it calls $A$ after $T_i$, it returns not $x$ .

and we want that if A never switches, then all non-faulty nodes output x. If A does switch then all non-faulty nodes should output the same thing, but it can be 0 or 1. 

Then this is also impossible, even for one faulty node, which just goes offline. Note that this generalises Byzantine agreement, since if we could each node $i$ could call $A$ once at the start and use the output as $v_i$. (For the multi-valued case, we will define the problem so that this reduction does not hold.)

\begin{proof}[Proof sketch] We follow the notation of  \cite{flp} and assume for a contradiction that we use a correct protocol. 
Let $r$ be a run of the protocol where $A$ gives $0$ all the time.
Then by correctness $r$ decides $0$. Now we consider what can happen when $A$ switches to $1$ after each configuration in $r$. If it switches to $1$ at the start, then the protocol decides $1$.
If we switch to $1$ when all node have already decided $0$, then we decide $0$.

We claim that some configuration in the run $r$, where there are two runs from it where $A$ is always $1$ that decide $0$ and $1$. We call such states $1$-bivalent.
To see this, assume for a contradiction that $r$ contains no such configurations. Then there are successive configurations $C$,$C'$ such that if $A$ return $1$ in the future from $C$ then we always decide $0$ but from $C'$, we always decide $1$.
Let events be $(p,m,x)$ where node (processor/voter) $p$ receives message $m$ (which may be null) and executes some code where any calls to A return $x$ in $\{0,1\}$, then sends some messages. 
Then there is some event $(p,m,0)$ that when applied to $C$ gives $C'$. Now suppose that $p$ goes offline at $C$, then if $A$ always returns $1$ afterwards, then we still decide $1$. Thus there is a run $r'$ that starts at $C$ where $p$ takes no steps, $A$ always returns $1$ and all other nodes still output $1$.
But since $p$ takes no steps in $r'$, we can apply $r'$ after $(p, m, 0)$ and so we have that $C'$ has a run where $A$ always returns $1$ but decides $1$, which is a contradiction.

Now let $C$ be a $1$-bivalent configuration. We can follow the FLP proof to show that there is a run from $C$ for which $A$ always returns $1$, all messages are delivered but all configurations are 1-bivalent and so the protocol never decides. This completes the proof by contradiction that there is no correct protocol.
\end{proof}

\subsection{Definition of a Finality Gadget}

In this section we show how to extend the one-shot agreement to agreeing on a chain of blocks. One difficulty in formalising the problem is that the block production mechanism cannot be entirely separate from the finality gadget. In order to finalise new blocks, we must first build on the chain we have already finalised. So at a minimum, the block production mechanism needs to recognise which blocks the finality gadget has finalised. We will also allow the block production mechanism to interact with the state of the finality gadget in other ways.

We want the finality gadget to work with the most general block production mechanisms as possible. Thus we need a condition that combines the property of eventual consensus and this requirement to build on the last finalised block, but is otherwise not too restrictive.
We assume a kind of conditional eventual consensus.
If we keep building on our last finalised block $B$ and don't finalise any new blocks, then eventually we have consensus on a longer chain than just $B$, which the finality gadget can use to finalise another block.
We also want a protocol that does not terminate, but instead keeps on finalising more blocks. 

We assume that there is a block production protocol $P$ that runs at the same time as the finality gadget protocol $G$. Actors who are participants in both protocols may behave differently in $P$ depending on what happened in $G$.
However in the reverse direction, the only way that an honest voter $v$'s behaviour in $G$ is affected by $P$ is through a voting rule, a function $A(v,s_v,B)$ that depends on $v$ and its state $s_v$ and takes a block $B$ and returns a block $B'$ at the head of a chain including $B$.

We say that the system $G$, $P$, and $A$ achieves {\em conditional eventual consensus}, if $G$ has finalised a block $B$, then eventually, either $G$ will finalise some descendant of $B$ or else all the chains with head $A_{v,s_v}(B)$ for all voters $v$ at all future states $s_v$ will contain the same descendant $B'$ of $B$.

\begin{definition} \label{def:finality-gadget} 
Let $F$ be a protocol  with a set of voters $V$, a constant fraction of which may be Byzantine.
We say that $F$ solves {\em blockchain Byzantine finality gadget problem} if for every block production protocol $P$ and voting rule $A$ we have the following


\begin{itemize}
\item{\bf Safety:} All honest voters finalise the same block at each block number.
\item{\bf Liveness:} If the system $F, G, A$ achieves conditional eventual consensus, then all honest voters keep finalising blocks.
\item{\bf Validity:} If an honest voter finalises a block $B$ then that block was seen in the best chain observed by some honest voter containing some previously finalised ancestor of $B$,
\end{itemize}

\end{definition}


As an example, we could assume $F$ uses proof of work to build on the longest chain and includes the last block $G$ finalised. Then we take $A(v,s_v,B)$ as being the longest chain which includes $B$ and which $v$ sees in state $s_v$. It is well-known \cite{nakamoto08bitcoin} that longest chain with proof of work achieves eventual consensus under the right assumptions and similar arguments show that in this case we have conditional eventual consensus.
As long as we do not change the chain we are building on by finalising another block, we will eventually agree on some prefix longer than the last finalised block.
Thus, any finality gadget that satisfies Definition \ref{def:finality-gadget} will work in this system so that all honest voters finalise an increasingly long common chain.
Thanks to the abstraction above, we can switch $F$ for one of many possible alternative consensus algorithms and $G$ will still work.

\subsection{Preliminaries} \label{sec:prelims}

\paragraph{Network model}: We will be using the partially synchronous network model introduced by \cite{DLS} and in particular the gossip network variant used in \cite{Tendermint}.
We assume that any message sent or received by an honest participant reaches all honest participants within time $T$, but possibly only after some Global Synchronisation Time $\GST$.
Concretely, any message sent or received by some honest participant at time $t$ is received by all honest participants by time $\GST+T$ at the latest.

\paragraph{Voters:} \com{We will want to change the set of participants who actively agree sometimes. 
To model this, we have a large set of participants who follow messages.}
For each voting step, there is a set of $n$ voters.
We will frequently need to assume that for each such step, at most  $f < n/3$ voters are Byzantine.
We need $n-f$ of voters to agree on finality. Whether or not block producers ever vote, they will need to be participants who track the state of the protocol.

\paragraph{Votes:}A vote is a block hash, together with some metadata such as round number and the type of vote, such as {\em prevote} or {\em precommit}, all signed with a voter's private key.

\paragraph{Rounds:}Each participant has their own idea of what is the current round number. Every prevote and precommit has an associated round number. Honest voters only vote once (for each type of vote) in each round and do not vote in earlier rounds after later ones.
Participants need to keep track of which block they see as currently being the latest finalised block and an estimate of which block could have been finalised in the last round.

For block $B$, we write $\mathrm{chain}(B)$ for the chain whose head is $B$. The block number, $n(B)$ of a block $B$ is the length of $\mathrm{chain}(B)$.
For blocks $B'$ and $B$, we say $B$ is later than $B'$ if it has a higher block number.
We write $B > B'$ or that $B$ is descendant of $B'$ for $B$, $B'$ appearing in the same blockchain with $B'$ later i.e. $B' \in \mathrm{chain}(B)$ with $n(B) > n(B')$.
$B \geq B'$ and $B \leq B'$ are similar except allowing $B = B'$.
We write $B \sim B'$ or $B$ and $B'$ are on the same chain if $B<B'$, $B=B'$ or $B> B'$; and $B \nsim B'$ or $B$ and $B'$ are not on the same chain if there is no such chain.

Blocks are ordered as a tree with the genesis block as root. So any two blocks have a common ancestor but two blocks not on the same chain do not have a common descendant.
A vote $v$ for a block $B$ by a voter $V$ is a message signed by $V$ containing the blockhash of $B$ and meta-information like the round numbers and the type of vote. 

A voter equivocates in a set of votes $S$ if they have cast multiple different votes in $S$. We call a set $S$ of votes safe if the number of voters who equivocate in $S$ is at most $f$. We say that $S$ has a supermajority for a block $B$ if the set of voters who either have a vote for blocks $\geq B$ or equivocate in $S$ has size at least $(n+f+1)/2$.  We count equivocations as votes for everything so that observing a vote is monotonic, meaning that if $S \subset T$ then if $S$ has a supermajority for $B$ so does $T$, while being able to ignore yet more equivocating votes from an equivocating voter.

For our finality gadget (GRANDPA) we use the ghost~\cite{lewenberg15inclusive} eventual consensus algorithm as $F$.
The $2/3$-GHOST function $g(S)$ takes a set $S$ of votes and returns the block $B$ with highest block number such that $S$ has a supermajority for $B$.
If there is no such block, then it returns `nil`. \com{(if $f \neq \lfloor (n-1)/3 \rfloor$, then this is a misnomer and we may change the name of the function accordingly.)}
Note that, if $S$ is safe, then we can compute $g(S)$ by starting at the genesis block and iteratively looking for a child of our current block with a supermajority, which must be unique if it exists. Thus we have:
\begin{lemma} \label{lem:ghost-monotonicity}
Let $T$ be a safe set of votes. Then
\begin{enumerate}
\item The above definition uniquely defines $g(T)$
\item If $S \subseteq T$ has $g(S) \neq$ nil, then $g(S) \leq g(T)$.
\item If $S_i \subseteq T$ for $1 \leq i \leq n$ then all non-nil $g(S_i)$ are on a single chain with head $g(T)$.
\end{enumerate}

\end{lemma}

Note that we can easily update $g(S)$ to $g(S \cup \{v\})$, by checking if any child of $g(S)$ now has a supermajority.
The third rule tells us that even if participants see different subsets of the votes cast in a given voting round, this rule may give them different blocks but all such blocks are in the same chain under this assumption. 

Next, we define a notion of possibility to have a supermajority which says that if the set of all votes in a vote $T$ is safe and some participant observes a subset $S \subseteq T$ that has a supermajority for a block $B$ then all participants who see some other subset $S' \subseteq T$ still see that it is possible for $S$ to have a supermajority for $B$. We need a definition that extends to unsafe sets.
We say that it is {\em impossible} for a set $S$ to have a supermajority for $B$ if at least $(n+f+1)/2$ voters either vote for a block $\not \geq B$ or equivocate in $S$. Otherwise it is {\em possible} for $S$ to have a supermajority for $B$.

Note that if $S$ is safe, it is possible for $S$ to have a supermajority for $B$ if and only if there is a safe $T \supseteq S$ that has a supermajority for $B$, which can be constructed by adding a vote from $B$ for all voters without votes in $S$ and enough voters who already have votes in $S$ to bring the number of equivocations up to $f$.

We say that it is {\em impossible} for any child of $B$ to have a supermajority in $S$ if $S$ has votes from at least $2f+1$ voters and it is impossible for $S$ to have a supermajority for each child of $B$ appearing on the chain of any vote in $S$.
Again, provided $S$ is safe, this holds if and only if for any possible child of $B$, there is no safe $T \subseteq S$ that has a supermajority for that child.
\com{
Note that it is possible for an unsafe $S$ to both have a supermajority for $S$ and for it to be impossible to have such a supermajority under these definitions, as we regard such sets as impossible anyway.
}
\begin{lemma} \label{lem:impossible}
\begin{itemize}
\item[(i)] If $B' \geq B$ and it is impossible for $S$ to have a supermajority for $B$, then it is impossible for $S$ to have a supermajority for $B'$.
\item[(ii)] If $S \subseteq T$ and it is impossible for $S$ to have a supermajority for $B$, then it is impossible for $T$ to have a supermajority for $B$.
\item[(iii)] If $g(S)$ exists and $B \nsim g(S)$ then it is impossible for $S$ to have a supermajority for $B$.
\end{itemize}
\end{lemma}

\section{Finality Gadget Protocols} \label{sec:finality}

In order to find a solution to the finality gadget protocol we look in
consensus protocols that solve the stronger problem as described in the previous section. The key idea for our solution is to inherit the safety properties of a consensus protocol, but use the underlying blockchain as the driving force of liveness. This results in a protocol which does not stop when for example the network is split. 
Instead, only the finalization stops, but the blocks keep getting created and propagated to everyone.
This means that when the conditions are safe again, the finality gadget only needs to finalize the head of the chain\footnote{Which the oracle will return quickly to a supermajority of miners.},
instead of having to transmit and run consensus on every block. 

\subsection{The GRANDPA Protocol}\label{sec:grandpa}
In this section, we give our solution to the Byzantine finality gadget problem, GRANDPA.  Our finality gadget works the partially synchronous setting, we also provide a fully asynchronous solution in Appendix~\ref{app:async}.

GRANDPA works in rounds, each round has a set of $3f+1$ eligible voters, $2f+1$ of which are assumed honest. Furthermore, we assume that each round has a participant designated as primary and all participants agree on the voter sets and primary. We will can either choose the primary pseudorandomly from or rotate through the voter set.
On a high-level, each round consists of a double-echo protocol after which every party waits in order to detect whether we can finalize a block in this round (this block does not need to be the immediate ancestor of the last finalized block, it might be far ahead from the last finalized block). If the round is unsuccessful, the parties simply move on to the next round with a new primary. When a good primary is selected, the oracle is consistent (returns the same value to all honest parties),
and the network is in synchrony (after $\GST$), then a new block will be finalized and it will transitively finalized all its ancestors.

More specifically, we let $V_{r,v}$ and $C_{r,v}$ be the sets of prevotes and precommits respectively received by $v$ from round $r$ at the current time.

We define $E_{r,v}$ to be $v$'s estimate of what might have been finalised in round $r$, given by the last block in the chain with head $g(V_{r,v})$ for which it is possible for $C_{r,r}$ to have a supermajority. Next we define a condition which will allow us to safely conclude that $E_{r,v} \geq B$ for all $B$ that might be finalised in round $r$:
If either $E_{r,v} < g(V_{r,v})$ or it is impossible for $C_{r,v}$ to have a supermajority for any children of $g(V_{r,v})$, then we say that {\em $v$ sees that round $r$ is completable}.  $E_{0,v}$ is the genesis block, assuming we start at $r=1$.  

In other words, a round $r$ is completable when our estimate chain $E_{r,v}$ contains everything that could have been finalised in round $r$, which makes it possible to begin the next round $r+1$.

We have a time bound $T$ that after $\GST$ suffices for all honest participants to communicate with each other.
Inside a round, the properties both of $E_{r,v}$ having a supermajority, meaning $E_{r,v} < g(V_{r,v})$, as well as of it being impossible to have a supermajority for some given block are monotone, so the property of being completable is monotone as well.
We therefore expect that, if anyone sees a round is completable, then everyone will see this within time $T$. Leaving a gap of $2T$ between steps is then enough to ensure that every party receives all honest votes before continuing.

\paragraph{Protocol Description.}
In round $r$ an honest participant $v$ does the following:

\noindent \fbox{\parbox{6.3in}{

\begin{enumerate}
\item A voter $v$ can start round $r > 1$ when round $r-1$ is completable and $v$ has cast votes in all previous rounds where they are a voter. Let $t_{r,v}$ be the time $v$ starts round $r$.

\item At time $t_{r,v}$, if $v$ is the primary of this round and has not finalised $E_{r-1,v}$ then they broadcast $E_{r-1,v}$. If they have finalised it, they can broadcast $E_{r-1,v}$ anyway (but do not need to).

\item If $v$ is a voter for the prevote of round $r$, $v$ waits until either it is at least time $t_{r,v}+2T$ or round $r$ is completable, then broadcasts a prevote.
They prevote for the head of the best chain containing $E_{r-1,v}$ unless we received a block $B$ from the primary and $g(V_{r-1,v}) \geq B > E_{r-1,v}$, in which case they use the best chain containing $B$ instead.

\item If $v$ is a voter for the precommit step in round $r$, then they wait until $g(V_{r,v}) \geq E_{r-1,v}$ and one of the following  conditions holds
\begin{itemize}
\item[(i)] it is at least time $t_{r,v}+4T$, 
\item[(ii)] round $r$ is completable or
\item[(iii)] it is impossible for $V_{r,v}$ to have a supermajority for any child of $g(V_{r,v})$,
\end{itemize}
and then broadcasts a precommit for $g(V_{r,v})$ {\em( (iii) is optional, we can get away with just (i) and (ii))}.

\end{enumerate}

}}

Note that $C_{r,v}$ and $V_{r,v}$ may change with time and also that $E_{r-1,v}$, which is a function of $V_{r-1,v}$ and $C_{r-1,v}$, can also change with time if $v$ sees more votes from the previous round.

\paragraph{Finalisation.}

If, for some round $r$, at any point after the precommit step of round $r$, we have that $B=g(C_{r,v})$ is later than our last finalised block and $V_{r,v}$ has a supermajority, then we finalise $B$.
We may also send a commit message for $B$ that consists of $B$ and a set of precommits for blocks $\geq B$ (ideally for $B$ itself if possible see "Alternatives to the last blockhash" below). 

To avoid spam, we only send commit messages for $B$ if we have not receive any valid commit messages for $B$ and its descendants and we wait some time chosen uniformly at random from $[0,1]$ seconds or so before broadcasting.
If we receive a valid commit message for $B$ for round $r$, then it contains enough precommits to finalise $B$ itself if we haven't already done so, so we'll finalise $B$ as long as we are past the precommit step of round $r$.

\section{ Analysis }

To analyse the performance of our finality gadget, we will need versions of our properties that appropriately depend on time: 

\begin{itemize}
\item{\bf Fast termination:} {\em If the last finalised block has number $n$ and, until another block is finalised, the best chain observed by all participants will include the same block with block number $n+1$, then a block with number $n+1$ will be finalised within time $T$.}
\item{\bf Recent validity:} {\em If an honest voter finalises a block $B$ then that block was seen in the best chain observed by some honest voter containing some previously finalised ancestor of $B$ more recently than time $T$ ago.}
\end{itemize}

Intuitively, fast termination implies that we finalise blocks fast as long as the block production mechanism achieves consensus fast whereas recent validity bounds the cost of starting to agree on something the block production mechanism's consensus later decides is not the best. In this case, we may waste time building on a chain that is never finalised so it is important to bound how long we do that.

These properties will typically only hold with high probability. In the asynchronous case, we would need to measure time in rounds of the protocol rather than seconds to make sense of these properties.  We are also interested in being able to remove and punish Byzantine voters, for which we will need:

\begin{itemize}
	\item{\bf Accountable Safety:} {\em If blocks on different chains are finalised, then we can identify at least $f+1$ Byzantine voters.}
\end{itemize}

\subsection{ Accountable Safety}

The first thing we want to show is asynchronous safety, assuming we have at most $f$ Byzantine voters. This follows from the property that if $v$ sees round $r$ as completable then any block $B$ with $E_{r,v} \not\leq B$ has that it is impossible for one of $C_{r,v}$ or $V_{r,v}$ to have a supermajority for $B$ and so $B$ was not finalised in round $r$. This ensures that all honest prevotes and precommits in round $r+1$ are for chains that include any blocks that could have been finalised in round $r$. With an induction, this is what ensures that we cannot finalise blocks on different chains. To show accountable safety, we need to turn this proof around to show the contrapositive, when we finalise different blocks , then there are $f+1$ Byzantine voters. If we make this proof constructive, then it gives us a challenge procedure, that can assign blame to such voters.

\begin{theorem} \label{thm:accountable} If the protocol finalises any two blocks $B,B'$ for which valid commit messages were sent, but which do not lie on the same chain, then there are at least $f+1$ Byzantine voters who all voted in a particular vote. Furthermore, there is a synchronous procedure to find some such set $X$ of $f+1$ Byzantine voters.
\end{theorem}

The challenge procedure works as follows: If $B$ and $B'$ are committed in the same round, then the union of their precommits must contain at least $f$ equivocations, so we are done.  Otherwise, we may assume by symmetry that $B$ was committed in round $r$ and $B'$ in round $r' > r$.  There are at least $n-f$ voters who precommitted $\geq B'$ or equivocated in round $r$ in their commit messages, so we ask those who precommitted $\geq B'$ why they did so.

Starting with $r''=r' $, we ask queries of the following form: 
\begin{itemize}
\item Why was $E_{r''-1} \not\geq B$  when you prevoted for or precommitted to $B'' \not\geq B$ in round $r'' > r$?
\end{itemize}
\noindent Any honest voter should be able to respond to this, as is shown in Lemma \ref{lem:honest-answer} below. 

The response is of the following form:
\begin{itemize}
\item A either a set $S$ of prevotes for round $r''-1$, or else a set $S$ of precommits for round $r''-1$, in either case such that it is impossible for $S$ to have a supermajority for $B$.
\end{itemize}

Any honest voter should respond.  In particular, if no voter responds, then we consider all  voters how should have responded but didn't as Byzantine and we return this set of voters, along with any equivocators, which will be at least $n-f$ voters total. If any do respond, then if $r'' > r+1$, we can ask the same query for at least $n-f$ voters in round $r''-1$. We note however that if any voters do respond then we will not punish non-responders.

If we ask such queries for a vote in all rounds between $r''=r'$ and $r''=r+1$ and get valid responses, since some voter responds when $r''=r+1$, then we have either a set $S$ of prevotes or precommits in round $r$ that show it is impossible for $S$ to have a supermajority for $B$ in round $r$.

If $S$ is a set of precommits, then if we take the union of $S$ and the set of precommits in the commit message for $B$, then the resulting set of precommits for round $r$ has a supermajority for $B$ and it is impossible for it to have a supermajority for $B$. This is possible if the set is not safe and so there must be at least $f+1$ voters who equivocate an so are Byzantine.

If we get a set $S$ of prevotes for round $r$ that does not have a supermajority for $B$, then we need to ask a query of the form

\begin{itemize}
\item Which prevotes for round $r$ have you seen?
\end{itemize}
\noindent to all the voters of precommit in the commit message for $B$  who voted for blocks $B'' \geq B$. There must be $n-f$ such voters and a valid response to this query is a set $T$ of prevotes for round $r$ with a supermajority for $B''$ and so a supermajority for $B$.

If any give a valid response, by a similar argument to the above, $S \cup T$ will have $f+1$ equivocations.

So we either discover $f+1$ equivocations in a vote or else $n-f > f+1$ voters either equivocate or fail to validly respond like a honest voter could do to a query.

\begin{lemma} \label{lem:honest-answer}
An honest voter can answer the first type of query.
\end{lemma}
We first show that, if a prevote or precommit in round $r$ is cast by an honest voter $v$ for a block $B''$, then at the time of the vote we had $B'' \geq E_{r-1,v}$.
Prevotes should be for the head of a chain containing either $E_{r-1,v}$ or some $B''' > E_{r-1,v}$ by step 2 or 3.  In either case we have $B'' \geq E_{r-1,v}$. Precommits should be for $g(V_{r,v})$ but $v$ waits until $g(V_{r,v}) \geq E_{r-1,v}$, by step 4, before precommitting, so again this holds.
It follows that, if $B'' \not\geq B$, then we had $E_{r-1,v} \not\geq B$.

We next show that if we had $E_{r-1,v} \not\geq B$ at the time of the vote then we can respond to the query validly, by demonstrating the impossibility of a supermajority for $B$. 
If $B$ was not on the same chain with $g(V_{r-1,v})$, then by Lemma \ref{lem:impossible} (iii), it was impossible for $V_{r-1,v}$ to have a supermajority for $B$, as desired. 
If $B$ was on the same chain as $g(V_{r-1,v})$, then it was on the same chain as $E_{r-1,v}$ as well.  In this case, we must have $B > E_{r-1,v}$ since $E_{r-1,v} \not\geq B$.
 However, possibly using that round $r-1$ is completable, it was impossible for $C_{r-1,v}$ to have a supermajority for any child of $E_{r-1,v}$ on the same chain with $g(V_{v,r})$ and in particular for the child of $E_{r-1,v}$ on $\textrm{chain}(B)$.
By Lemma \ref{lem:impossible} (i), this means $C_{r-1,v}$ did not have a supermajority for $B$, again as desired.

Thus we have that, at the time of the vote, for one of $V_{r-1,v}$, $C_{r-1,v}$, it was impossible to have a supermajority for $B$. The current sets $V_{r-1,v}$ and $C_{r-1,v}$ are supersets of those at the time of the vote, and so by Lemma \ref{lem:impossible} (ii), it is still impossible. Thus $v$ can respond validly.

This is enough to show Theorem \ref{thm:accountable}.  Note that if $v$ sees a commit message for a block $B$ in round $r$ and has that $E_{r',v} \not\geq B$, for some completable round $r' \geq r$, then they should also be able to start a challenge procedure that successfully identifies at least $f+1$ Byzantine voters in some round. Thus we have that:

\begin{corollary} \label{cor:overestimate-final}
If there at most $f$ Byzantine voters in any vote, $B$ was finalised in round $r$, and an honest participant $v$ sees that round $r' \geq r$ is completable, then $E_{r',v} \geq B$.
\end{corollary}

\subsection{Liveness }

We show the protocol is deadlock free and also that it finalises new blocks quickly in a weakly synchronous model.
For this section, we will assume that there are at most $f < n/3$ Byzantine voters for each vote, and so that the sets of prevotes and precommits for each round are safe.

We define $V_{r,v,t}$ be the set $V_{r,v}$ at time $t$ and similarly for $C_{r,v,t}$ and the block $E_{r,v,t}$ .

We first show that the completability of a round and the estimate for a completable round are monotone in the votes we see, in the latter case monotonically decreasing: 

\begin{lemma} \label{lem:message-monotonicity-completed-estimate}
Let $v,v'$ be (possibly identical) honest participants, $t,t'$ be times, and $r$ be a round.
Then if $V_{r,v,t} \subseteq V_{r,v',t'}$ and $C_{r,v,t} \subseteq C_{r,v',t'}$and $v$ sees that $r$ is completable at time $t$, then $E_{r,v',t'} \leq E_{r,v,t}$ and $v'$ sees that $r$ is completable at time $t'$.
\end{lemma}

\begin{proof} 
Since $v$ sees that $r$ is completable at time $t$, 
either $E_{r,v} < g(V_{r,v})$ requiring $(n+f+1)/2 > 2f + 1$ votes, or else it is impossible for $C_{r,v}$ to have a supermajority for any children of $g(V_{r,v})$, requiring $2f + 1$ votes.
In either case, both $V_{r,v,t}$ and $C_{r,v,t}$ contain votes from $2f + 1$ voters and so the same holds for $V_{r,v',t'}$ and $C_{r,v',t'}$. 
By Lemma \ref{lem:ghost-monotonicity} (ii), $g(V_{r,v',t'}) \geq g(V_{r,v,t})$.
As it is impossible for $C_{r,v,t}$ to have a supermajority for any children of $g(V_{r,v,t})$, it follows from Lemma \ref{lem:impossible} (i \& ii) that it is impossible for $C_{r,v',t'}$ as well, and so both $E_{r,v',t'} \leq g(V_{r,v,t})$ and $v'$ sees $r$ is completable at time $t'$.
But now $E_{r,v,t}$ and $E_{r,v',t'}$ are the last blocks on $\textrm{chain}(g(V_{r,v,t}))$ for which it is possible for $C_{r,v,t}$ and $C_{r,v',t'}$ respectively to have a supermajority, 
As it is possible for $C_{r,v',t'}$ to have a supermajority for $E_{r,v',t'}$, then it is possible for $C_{r,v,t}$ to have a supermajority for $E_{r,v',t'}$ as well, by Lemma \ref{lem:impossible} (ii) and tolerance assumptions, so $E_{r,v',t'} \leq E_{r,v,t}$.
\end{proof}
 
\subsubsection{Deadlock Freeness}

Now we can show deadlock freeness for the asynchronous gossip network model, when a message that is sent or received by any honest participant is eventually received by all honest participants.

\begin{proposition} Suppose that we are in the asynchronous gossip network model and that at most $f$ voters for any vote are Byzantine. Then the protocol is deadlock free.\end{proposition}

\begin{proof} We need to show that if all honest participants reach some vote, then all of them eventually reach the next.

If all honest voters reach a vote, then they will vote and all honest participants see their votes. We need to deal with the two conditions that might block the algorithm even then.
To reach the prevote of round $r$, a participant may be held up at the condition that round $r-1$ must be completable. To reach the precommit, a voter may be held up by the condition that $g(V_{r,v}) \geq E_{r-1,v}$.

For the first case, the prevote, let $S$ be the set of all prevotes from round $r-1$ that any honest voter saw before they precommitted in round $r-1$.
By Lemma \ref{lem:ghost-monotonicity}, when voter $v'$ precommitted, they do it for block $g(V_{r-1,v'}) \leq g(S)$.
Let $T$ be the set of precommits in round $r$ cast by honest voters.
Then for any block $B \not\leq g(S)$, $T$ does not contain any votes that are $\geq B$ and so it is impossible for $T$ to have a supermajority for $B$.
In particular, it is impossible for $T$ to have a supermajority for any child of $g(S)$. 

Now consider a voter $v$. By our network assumption, there is a time $t$ by which they have seen the votes in $S$ and $T$. Consider any $t' \geq t$.
At this point we have $g(V_{r,v,t;}) \geq g(S)$. It is impossible for $C_{r,v,t'}$ to have a supermajority for any child of $g(S)$ and so $E_{r-1,v,t'} \leq g(S)$, whether or not this inequality is strict, we satisfy one of the two conditions for $v$ to see that round $r-1$ is completable at time $t'$.
Thus if all honest voters reach the precommit vote of round $r-1$, all honest voters reach the prevote of round $r$.

Now we consider the second case, reaching the precommit. 
Note that any honest prevoter in round $r$ votes for a block $B_v \geq E_{r-1,v,t_v}$ where $t_v$ is the time they vote. Now consider any honest voter for the precommit $v'$. By some time $t'$, they have received all the messages received by each honest voter $v$ at time $t_v$ and $v'$'s prevote. 
Then by Corollary \ref{cor:overestimate-final}, $B_v \geq E_{r-1,v,t_v} \geq E_{r-1,v',t'}$. Since $V_{r,v',t'}$ contains these $B_v$, $g(V_{r,v',t'}) \geq  E_{r-1,v',t'}$. Thus if all honest voters prevote in round $r$, eventually all honest voters precommit in round $r$.

An easy induction completes the proof of the proposition.
\end{proof}

\subsubsection{Weakly synchronous liveness}

Now we consider the weakly synchronous gossip network model. The idea that there is some global stabilisation time($\GST$) such that any message received or sent by an honest participant at time $t$ is received by all honest participants at time $\max\{t,\GST\}+T$.

Let $t_r$ be the first time any honest participant enters round $r$ i.e. the minimum over honest participants $v$ of $t_{r,v}$.

\begin{lemma} \label{lem:timings}
Assume the weakly synchronous gossip network model and that each vote has at most $f$ Byzantine voters. Then if $t_r \geq \GST$, we have that
\begin{itemize}
\item[(i)] $t_r \leq t_{r,v} \leq t_r+T$ for any honest participant $v$,
\item[(ii)] no honest voter prevotes before time $t_r+2T$,
\item[(iii)] any honest voter $v$ precommits at the latest at time $t_{r,v}+4T$,
\item[(iv)] for any honest $v$, $t_{r+1,v} \leq t_r + 6T$.
\end{itemize}
\end{lemma}

\begin{proof} Let $v'$ be one of the first honest participants to enter round $r$ i.e. with $t_{r,v'}=t_r$. 
By our network assumption, all messages received by $v'$ before they ended are received by all honest participants before time $t_r+T$.
In particular at time $t_r$, $v'$ sees that all previous rounds are completable and so by Corollary \ref{cor:overestimate-final}, so does every other honest participant by time $t_r+T$.
Also since for $r' < r$, at some time $s_{r'} \leq t_r$ $g(V_{r',v',s_r'}) \geq E_{r',v',s_r'}$, again by Lemma 4, for all honest $v$, $g(V_{r',v,t_r+T}) \geq E_{r',v,t_r+T}$. Looking at the conditions for voting, this means that any honest voter does not need to wait before voting in any round $r' \leq r$. 
Thus they cast any remaining votes and enter round $r$ by time $t_r + T$. This shows (i).

For (ii), note that the only reason why an honest voter would not wait until time $t_{r,v}+2T \geq t_r+ 2T$ is when $n-f$ voters have already prevoted. But since some of those $n-f$ votes are honest, this is impossible before $t_r+2T$

Now an honest voter $v''$ prevotes at time $t_{r,v''}+2T \leq t_r +3T$ and by our network assumptions all honest participants receive this vote by time $t_r+4T$. An honest voter for the precommit $v$ has also received all messages that $v''$ received before they prevoted by then.
Thus the block they prevoted has $B_{v''} \geq E_{r-1,v''} \geq E_{r-1,v,t_r+4T}$, since this holds for every honest voter $v''$, $g(V_{r,v,t_r+4T}) \geq E_{r-1,v,t_r+4T}$. Thus they will precommit by time $t_{r,v}+4T$ which shows (iii).

By the network assumption an honest voter $v'$'s precommit will be received by all honest participants $v$ by time $t_{r,v'}+ 5T \leq t_r+6T$.
Since $v$ will also have received all prevotes $v$ say when they precommitted by this time, their vote $B_{v'}$ will have $B_{v'}=g(V_{r,v'}) \leq g(V_{r,v,t_r+6T})$.
Thus $C_{r, v, t_r+6T}$ contains precommits from $n-f$ voters $v'$ with $B_{v'} \leq g(V_{r,v,t_r+6T})$ and thus it is impossible for $C_{r,v,t_r+6T}$ to have a supermajority for any children of $g(V_{r,v, t_r+6T})$.
Thus $v$ sees that round $r$ is completable at time $t_r+6T$. Since they have already prevoted and precommitted if they were a voter, they will move to round $r+1$ by at latest $t_t+6T$. This is (iv).
\end{proof}

\begin{lemma} \label{lem:honest-prevote-timings}
Suppose $t_r \geq \GST$ and very vote has at most $f$ Byzantine voters. Let $H_r$ be the set of prevotes ever cast by honest voters in round $r$. Then
\begin{itemize}
\item[(a)] any honest voter precommits to a block $\geq  g(H_r)$,

\item[(b)] every honest participant finalises $g(H_r)$ by time $t_r+6T$.
\end{itemize}
\end{lemma}

\begin{proof} For (a), we separate into cases based on which of the conditions (i)-(iii) that we wait for to precommit hold.

For (i), all honest voters prevote in round $r$ by time $t_r+3T$. So any honest voter $v$ who precommits at or after time $t_{r,v}+4T \geq t_r+4T$ has received all votes in $H_r$ and by Lemma \ref{lem:ghost-monotonicity}, precommits to a block $\geq g(H_r)$.

For (ii), we argue that no honest voter commits a block $\not\geq g(H_r)$ first. The result will then follow by an easy induction once the other cases are dealt with. Suppose that no honest voter has precommitted a block $\not \geq g(H_r)$ so far and that a voter $v$ votes early because of (ii).

Note that, since we assume that all precommits by honest voters so far were $\geq g(H_r)$, it is possible for $C_{r,v}$ to have a supermajority for $g(H_r)$.
For (ii) to hold for a voter $v$ i.e for round $r$ to be completable, it must be the case that either it is impossible for $C_{r,v}$ to have a supermajority for $g(V_{r,v})$ or else be impossible for $C_{r,v}$ to have a supermajority for any children of $g(V_{r,v})$. By Lemma \ref{lem:impossible} cannot have $g(V_{r,v}) < g(H_r)$.
But by Lemma \ref{lem:ghost-monotonicity}, these are on the same chain and so $g(V_{r,v}) \geq g(H_r)$. Since this is the block $v$ precommits to, we are done in case (ii)

For (iii), let $v$ be the voter in question. Note that since $n-f$ honest voters prevoted $\geq g(H_r)$, it is possible for $V_{r,v}$ to have a supermajority for $g(H_r)$. By Lemma \ref{lem:ghost-monotonicity}, $g(V_{r,v})$ is on the same chain as $g(H_r)$.
For (iii), it is impossible for $V_{r,v}$ to have a supermajority for any children of $g(V_{r,v})$. If we had $g(V_{r,v}) < g(H_r)$, by Lemma \ref{lem:impossible}, this would mean that it would be impossible for $V_{r,v}$ to have a supermajority for $g(H_r)$ as well. So it must be that $g(V_{r,v} )\geq g(H_r)$ as required.

For (b), combining (a) and Lemma \ref{lem:timings} (iii), we have that any honest voter $v$ precommits $\geq g(H_r)$ by time $t_{r,v}+4T$. By our network assumption, all honest participants receive these precommits by time $t_r+6T$ and so finalise $g(H_r)$ if they have not done so already.
\end{proof}

\begin{lemma} \label{lem:primary-finalises}
 Suppose that $t_r \geq \GST$, the primary $v$ of round $r$ is honest and no vote has more than $f$ Byzantine voters. Let $B=E_{r-1,v,t_{v,r}}$ be the block $v$ broadcasts if it is not final. Then every honest prevoter prevotes for the best chain including $B$ and all honest voter finalise $B$ by time $t_r+6T$.
 \end{lemma}

\begin{proof} By Lemma \ref{lem:timings} and our network assumptions, no honest voter  prevotes before time $t_r+2T \geq t_{r,v}+2T$ and so at this time, they will have seen all prevotes and precommits seen by $v$ at $t_{r,v}$ and the block $B$ if $v$ broadcast it then. By Lemma \ref{lem:message-monotonicity-completed-estimate}, any honest voter $v'$ has $E_{r-1,v'} \leq B \leq g(V_{r-1,v})$ then.

So if the primary broadcast $B$, then $v'$ prevotes for the best chain including $B$. If the primary did not broadcast $B$, then they finalise it. By Corollary \ref{cor:overestimate-final}, it must be that $E_{r-1,v'} \geq B$ and so $E_{r-1,v'}=B$ and so in this case $v'$ also prevotes for the best chain including $B$.

Since all honest voters prevote $\geq B$, $g(H_r) \geq B$ and so by Lemma \ref{lem:honest-prevote-timings}, all honest participants finalise $B$ by time $t_r+6T$
\end{proof}

\begin{lemma}
 Suppose that $t_r \geq \GST+T$ and the primary of round $r$ is honest. 
Let $B$ be the latest block that is ever finalised in rounds  $<r$ (even if no honest participant finalises it until after $t_r$). If all honest voters for the prevote in round $r$ agree that the best chain containing $B$ include the same child $B'$ of $B$, then they all finalises some child of $B$ before $t_r+6T$.
\end{lemma}

\begin{proof} By Corollary \ref{cor:overestimate-final}, any honest participant sees that $E_{r-1} \geq B$ during round $r$. Let $v$ be the primary of round $r$ and $B''=E_{r-1,v,t_{r,v}}$. If $B'' > B$, then by Lemma \ref{lem:primary-finalises}, all honest participants finalise $B''$ by time $t_r+6T$ which means they finalised a child of $B$. If $B''=B$, then by Lemma \ref{lem:honest-prevote-timings}, all honest voters prevote for the best chain including $B$.
By assumption these chains include $B'$ and so $g(H_r) \geq B$. By Lemma \ref{lem:honest-prevote-timings}, this means that $B'$ is finalised by time $t_r+6T$.
\end{proof}

\subsubsection{Recent Validity}

\begin{lemma} \label{lem:honest-recent-validity}
Suppose that $t_r \geq \GST$, the primary of round $r$ is honest and all votes have at most $f$ Byzantine voters.
Let $B$ be a block that less than $f+1$ honest prevoters in round $r$ saw as being in the best chain of an ancestor of $B$ at the time they prevoted.
Then either all honest participants finalise $B$ before time $t_r+6T$ or no honest participant ever has $g(V_{r,v}) \geq B$ or $E_{r,v} \geq B$.
\end{lemma}

\begin{proof} Let $v'$ be the primary of round $r$ and let $B'=E_{r-1,v',t_{r,v'}}$. If $B' \geq B$, then by Lemma \ref{lem:primary-finalises}, all honest participants finalise $B$ by time $t_r+6T$. If $B' \not\geq B$, then by Lemma \ref{lem:primary-finalises}, at most $f$ honest voters prevotes $\geq B$. In this case, less than $2f+1 \leq (n+f+1)/2$ prevoters vote $\geq B$ or equivocate and so no honest participant ever has $g(V_{r,v}) \geq B$.
\end{proof}

\begin{corollary} For $t - 6T > t' \geq \GST$, suppose that an honest participant finalises $B$ at time $t$ but that no honest voter has seen $B$ as in the best chain containing some ancestor of $B$ in between times $t'$ and $t$, then at least $(t-t')/6T - 1$ rounds in a row had Byzantine primaries. \end{corollary}

\bibliography{grandpa} 

\end{document}